\documentclass[conference]{IEEEtran}
\IEEEoverridecommandlockouts
\usepackage[noadjust]{cite}
\usepackage{amsmath,amssymb,amsfonts}
\usepackage{algorithmic}
\usepackage{graphicx}
\usepackage{textcomp}
\usepackage{xcolor}
\def\BibTeX{{\rm B\kern-.05em{\sc i\kern-.025em b}\kern-.08em
    T\kern-.1667em\lower.7ex\hbox{E}\kern-.125emX}}

\usepackage[colorlinks=false,hidelinks]{hyperref}
\usepackage{cleveref}
\usepackage{braket}
\usepackage{tikz}
\usetikzlibrary{external}
\usetikzlibrary{quantikz2}
\usepackage[customcolors,shade]{hf-tikz}
\usepackage{caption}
\usepackage{subcaption}
\usepackage{todonotes}
\usepackage{amsthm}
\usepackage[T1]{fontenc}

\usepackage{pgfplots}
\pgfplotsset{width=10cm,compat=1.9}

\usepgfplotslibrary{external}
\newtheorem{theorem}{Theorem}
\newtheorem{corollary}{Corollary}

\DeclareMathOperator*{\tr}{Tr}
\DeclareMathOperator*{\argmin}{arg\,min}
\DeclareMathOperator*{\locc}{LOCC}
\DeclareMathOperator*{\cptn}{CPTN}

\newcommand{\ketbra}[2]{|#1\rangle\langle#2|}
\newcommand{\ketbras}[1]{\ketbra{#1}{#1}}

\begin{document}

\title{Cutting a Wire with\\Non-Maximally Entangled States
\thanks{This work was partially funded by the BMWK project \textit{EniQmA} (01MQ22007B) and \textit{SeQuenC} (01MQ22009B).}
}

\author{
\IEEEauthorblockN{
Marvin Bechtold, 
Johanna Barzen, 
Frank Leymann, and 
Alexander Mandl} 
\IEEEauthorblockA{Institute of Architecture of Application Systems\\
University of Stuttgart\\ 
Universitätsstraße 38, 70569 Stuttgart, Germany\\
Email: \{bechtold, barzen, leymann, mandl\}@iaas.uni-stuttgart.de}}

\maketitle

\begin{abstract}
Distributed quantum computing supports combining the computational power of multiple quantum devices to overcome the limitations of individual devices.
Circuit cutting techniques enable the distribution of quantum computations via classical communication.
These techniques involve partitioning a quantum circuit into smaller subcircuits, each containing fewer qubits.
The original circuit's outcome can be replicated by executing these subcircuits on separate devices and combining their results.
However, the number of circuit executions required to achieve a fixed result accuracy with circuit cutting grows exponentially with the number of cuts, posing significant costs.
In contrast, quantum teleportation allows the distribution of quantum computations without an exponential increase in circuit executions.
Nevertheless, each teleportation requires a pre-shared pair of maximally entangled qubits for transmitting a quantum state, and non-maximally entangled qubits cannot be used for this purpose.
Addressing this, our work explores utilizing non-maximally entangled qubit pairs in wire cutting, a specific form of circuit cutting, to mitigate the associated costs. 
The cost of this cutting procedure reduces with the increasing degree of entanglement in the pre-shared qubit pairs.
We derive the optimal sampling overhead in this context and present a wire cutting technique employing pure non-maximally entangled states that achieves this optimal sampling overhead.
Hence, this offers a continuum between existing wire cutting and quantum teleportation.

\end{abstract}

\begin{IEEEkeywords}
Distributed Quantum Computing, Circuit Cutting, Quantum Teleportation, Entanglement
\end{IEEEkeywords}

\section{Introduction}
Quantum computing holds immense potential for solving problems intractable for classical computers~\cite{Cao2019,Giani2021}. 
However, current quantum devices face various constraints, including a limited number of available qubits, and scaling up to larger devices remains a considerable challenge~\cite{Preskill2018}. 
One promising approach to address this limitation involves utilizing multiple smaller quantum devices that can exchange classical~\cite{Avron2021,Dunjko2018} or, in the near future, quantum information~\cite{Cuomo2020,Khait2023}.
These modular systems enable distributed quantum computation, offering a pathway to scalability~\cite{Bravyi2022}. 
Moreover, as quantum devices mature, the integration of multiple quantum devices into the computing infrastructure is anticipated~\cite{Furutanpey2023_QuantumEdge}, further emphasizing the importance of distributed quantum computation.
\looseness=-1

Circuit cutting is a technique used to distribute the computation of a quantum circuit over multiple devices that can exchange only classical information~\cite{Bravyi2016,Brenner2023,Mitarai2021,Peng2019,Piveteau2022}.
This technique involves decomposing a large quantum circuit into smaller subcircuits by cutting it, allowing for their execution on multiple smaller quantum devices.
The subcircuit results can then be recombined through classical postprocessing to reproduce the original circuit's outcome. 
One specific circuit cutting technique is wire cutting, where a wire in a quantum circuit, i.e., a qubit, is interrupted through a series of measurements and subsequent qubit initalizations.
This allows for partitioning the execution of a large circuit across multiple smaller devices by cutting the connecting wires.
However, wire cutting increases the number of circuit executions required to accurately reproduce the result of the original circuit~\cite{Brenner2023}.
This overhead in required shots scales exponentially with the number of cuts that are performed.
Therefore, minimizing the overhead of each cut is crucial to extend the applicability of circuit cutting.

Quantum teleportation is an alternative technique for distributing quantum computations across multiple devices~\cite{Cuomo2020}.
In contrast to circuit cutting, it eliminates the need for additional shots when executing a quantum circuit across multiple devices by leveraging maximally entangled states. 
Based on a pre-shared maximally entangled state, quantum teleportation allows the transmission of quantum information between quantum devices by exchanging only two classical bits~\cite{Bennett1993}. 
Although generating these entangled qubit pairs between multiple commercial quantum devices is currently unattainable, it is anticipated to become feasible in the near future~\cite{Bravyi2022}.
\looseness=-1

Quantum teleportation and circuit cutting represent two extremes in the amount of entanglement utilized between the two participating devices.
Teleportation relies on maximally entangled states, while circuit cutting operates without any entanglement. 
However, neither of these techniques allows utilizing states with intermediate levels of entanglement, known as \textit{non-maximally entangled~(NME)} states.
To close this gap, this work explores the potential of leveraging NME states with varying degrees of entanglement. 
The primary objective is to investigate how these states can reduce the overhead associated with wire cutting.
Building on our earlier research~\cite{Bechtold2023_CircuitCuttingEntanglement_preprint}, this work contributes in two significant ways. 
Firstly, we establish the optimal overhead for wire cutting using arbitrary NME states in \Cref{theorem_overhead}. 
Secondly, \Cref{theorem_decomposition} introduces a wire cut that attains this optimal overhead employing pure NME states and is used in experiments.

The remainder of this work is structured as follows:
\Cref{sec:preliminaries} establishes the preliminaries required for wire cutting with NME states as presented in \Cref{sec:main-section}.
Subsequently, in \Cref{sec:num_experiments}, we present numerical experiments conducted with our newly introduced wire cut.
\Cref{seq:related_work} provides a summary of the related work and \Cref{sec:conclusion} concludes the paper.

\section{Preliminaries}\label{sec:preliminaries}
This section explores fundamental concepts underpinning the application of NME states in a wire cutting procedure, as detailed in \Cref{sec:main-section}.
Initially, NME states are described, and their entanglement is quantified.
The concept of quasiprobability simulation is introduced, and its role in simulating maximally entangled states and wire cutting is outlined. 
Finally, we provide a concise overview of quantum teleportation, focusing on its use in transmitting quantum information via NME states.

\subsection{Quantifying entanglement in NME qubit pairs}
Entanglement is a fundamental quantum mechanical resource in bipartite quantum systems, governed by their joint Hilbert space $A \otimes B$~\cite{Nielsen2009}. 
The composite system's state is described by a density operator $\rho_{AB}$, a positive Hermitian operator on $A \otimes B$ with unit trace.
The set of density operators is referred to as $D(A\otimes B)$.
The subscript in $\rho_{AB}$ optionally specifies the associated Hilbert spaces.
The evolution and manipulation of density operators are described by superoperators, which are linear operators that act on density operators in $D(A\otimes B)$.
For a superoperator $\mathcal{E}$ to be physically realizable, it must be completely positive and trace-nonincreasing, i.e., $0 \le \tr[\mathcal{E}(\rho)] \le \tr[\rho]$ for any $\rho \in D(A\otimes B)$~\cite{Nielsen2009}.
The set of all \textit{completely positive trace-nonincreasing~(CPTN)} superoperators on the density operators in $D(A\otimes B)$ is denoted as $\cptn(A\otimes B)$.

In bipartite systems, an entangled state $\rho_{AB} \in D(A\otimes B)$ cannot be factorized into independent states of the subsystems.
Hence, no  $\rho_A \in D(A)$ and $\rho_B \in D(B)$ exists such that $\rho_{AB} = \rho_A \otimes \rho_B$. 
A maximally entangled state $\Phi_{AB}$ is characterized by its reduced density operators being maximally mixed~\cite{Goyeneche2015}, i.e. $\tr_A[\Phi_{AB}] = I/\dim(B)$ and $\tr_B[\Phi_{AB}] = I/\dim(A)$.
An important set of operators on entangled states is $\locc(A,B) \subset \cptn(A\otimes B)$, including only local transformations on the subsystems $A$ and $B$, coordinated through classical communication~\cite{Chitambar2014}.
An operator $\Lambda\in\locc(A, B)$ cannot increase the entanglement degree of a bipartite state $\rho_{AB}$.
Thus, the entanglement degree of the transformed state $\Lambda(\rho_{AB})$ remains at most equal to that of $\rho_{AB}$~\cite{Chitambar2014}.

Focusing on cutting a single wire, we consider entangled two-qubit states in the following.
Thus, $A$ and $B$ are two dimensional Hilbert spaces and we use $\Phi$ to refer to the maximally entangled two-qubit state given by $\Phi_{AB} = \ketbra{\Phi}{\Phi}_{AB}$ with $\ket{\Phi} = \frac{1}{\sqrt{2}}(\ket{00} + \ket{11})$.
To quantify the amount of entanglement of an arbitrary two-qubit NME state given by its density operator $\rho_{AB}\in D(A\otimes B)$, we asses its similarity to a pure maximally entangled two-qubit state.
This assessment considers transformations from the set $\locc(A,B)$ as they cannot increase entanglement, and all pure maximally entangled two-qubit states can be transformed into each other using operators from $\locc(A,B)$~\cite{Nielsen1999}.
Hence, the entanglement of an arbitrary two-qubit NME state $\rho_{AB}$ is quantified by its maximal overlap with the fixed maximally entangled state $\Phi_{AB}$, considering all possible transformations $\Lambda\in\locc(A, B)$.
Based on~\cite[Equation 11]{Yuan2023}, this is given by
\begin{align}\label{eq:overlap}
    f(\rho_{AB}) := \max_{\Lambda \in \locc(A,B)}\braket{\Phi_{AB}|\Lambda(\rho_{AB})|\Phi_{AB}}.
\end{align}
This measure is an entanglement monotone~\cite{Vidal2000,Verstraete2003}, i.e. it is nonincreasing for $\Lambda \in \locc(A,B)$:
\begin{align}
    f(\Lambda(\rho_{AB}))\le f(\rho_{AB}).
\end{align}
It ranges for a two-qubit state from $\frac{1}{2}$ denoting the absence of entanglement to $1$ indicating maximal entanglement.

In the following, we describe pure NME states used in our wire cutting procedure and derive an explicit expression for their maximal overlap with the maximally entangled state.
By making use of the Schmidt decomposition~\cite{Nielsen2009}, any entangled pure two-qubit state $\ket{\psi}$ can be expressed as 
\begin{align}\label{eq:schmidt}
    \ket{\psi} &= p_0 \ket{\xi_0} \ket{\zeta_0} + p_1 \ket{\xi_1} \ket{\zeta_1}\\
    &= p_0 \left(\ket{\xi_0} \ket{\zeta_0} + \frac{p_1}{p_0} \ket{\xi_1} \ket{\zeta_1}\right), \label{eq:schmidt2}
\end{align}
with $p_i \in \mathbb{R}_{\geq 0}$, $p_0 \ne 0$, and orthonormal sets of one-qubit states $\{\ket{\zeta_i}\}$ and $\{\ket{\xi_i}\}$.
It is worth noting that the Schmidt decomposition of any bipartite pure state can be computed on near-term quantum devices, given a sufficient number of instances of the state~\cite{BravoPrieto2020}.
The representation in \Cref{eq:schmidt2} shows that any two-qubit state can be reformulated as 
\begin{align}
    \ket{\psi} &= (U_A \otimes U_B) (\ket{\Phi^k}).
\end{align}
Herein, $U_A= \ketbra{\xi_0}{0} + \ketbra{\xi_1}{1}$ and $U_B = \ketbra{\zeta_0}{0} + \ketbra{\zeta_1}{1}$ are local unitary basis transformations from the computational basis into the basis of the Schmidt decomposition, and $\ket{\Phi^k}$ is a generalization of the Bell basis state $\ket{\Phi}$ given as 
\begin{align}\label{eq:nme}
    \ket{\Phi^{k}} = K(\ket{00} + k\ket{11})
\end{align}
with $K := \frac{1}{\sqrt{1+ k^{2}}}$ and $k \in \mathbb{R}_{\ge 0}$. 
Thus, $k=p_1/p_0$ and $K=p_0$.
We represent the density operator of this state as $\Phi^k = \ketbra{\Phi^k}{\Phi^k}$. 
Furthermore, since $f$ is nonincreasing under operators from $\locc(A, B)$ and $U = (U_A \otimes U_B)$ is a local unitary operator on $A$ and $B$, it holds that 
\begin{align}
    f(\psi_{AB}) &= f(U\Phi^k_{AB}U^{\dagger}) \le f(\Phi^k_{AB}) \\
    f(\Phi^k_{AB}) &= f(U^{\dagger}\psi_{AB}U) \le f(\psi_{AB}).
\end{align}
As a result, $f(\psi_{AB}) = f(\Phi^k_{AB})$.
Therefore, we solely focus in the following on NME qubit pairs of the form in \Cref{eq:nme}.

For the pure state $\Phi^k$, the maximal overlap with $\Phi$ is 
\begin{align}
    f(\Phi^k) &= \braket{\Phi|\Phi^k|\Phi}\\
    &= \frac{\left(k + 1\right)^{2}}{2 \left(k^{2} + 1\right)}. \label{eq:singlet_fraction_nme}
\end{align}
Further details on this calculation can be found in Appendix~\ref{sec:singlet_fraction}.
Hence, the state $\Phi^k$ is separable for $k=0$ and $k\to\infty$, while it is maximally entangled for $k=1$.

\subsection{Quasiprobability simulation of non-local operators}\label{sec:quasi_prob_sim}
As all circuit cutting techniques can be interpreted as quasiprobability simulations of non-local operators~\cite{Brenner2023}, we introduce the required concepts in this section.
The main idea is to simulate the outcome of a quantum circuit distributed across subsystems $A$ and $B$ by probabilistically replacing non-local operators between these subsystems with local operators and classical communiction~\cite{Piveteau2022}. 

To perform these replacements, a non-local operator \mbox{$\mathcal{E}\in \cptn(A\otimes B)$} is decomposed as follows~\cite{Piveteau2022}:
\begin{align}\label{eq:QPD}
    \mathcal{E} = \sum_{i=1}^{m} c_i \mathcal{F}_i.
\end{align}
Here, $\mathcal{F}_i\in \locc(A, B)$ and $c_{i}$  are real coefficients, which can take negative values but must satisfy the constraint $\sum_{i=1}^{m} c_{i} = 1$.
As the coefficients can also take negative values, this is referred to as a quasiprobability distribution, thereby coining the term \textit{quasiprobability decomposition (QPD)} of $\mathcal{E}$.

Using this QPD, the expectation value of the operator $\mathcal{E}$ on state $\rho$
concerning observable $O$ is expressed as~\cite{Brenner2023}
\begin{align}\label{eq:expectation}
\tr[O\mathcal{E}(\rho)] = \kappa\sum_{i=1}^{m} p_{i} \tr[O\mathcal{F}_{i}(\rho)] \operatorname{sign}(c_i), 
\end{align}
where $\kappa := \sum_i |c_i|$, $p_i := |c_{i}|/\kappa$, and $\operatorname{sign}(c_i)$ denotes the sign of $c_i$.
This formulation enables the computation of the expectation value using a Monte Carlo approach using operators from $\locc(A,B)$ only~\cite{Brenner2023}: for each shot, an index $i$ is selected at random with probability $p_{i}$, then the circuits associated with the operator $\mathcal{F}_{i}$ are executed, and the outcome according to $O$ is weighted by $\operatorname{sign}(c_i) \kappa$. 
The sum of the weighted results estimates the expectation value of the original operator $\mathcal{E}$.
Although this estimator preserves the expectation value, it increases the variance by $\kappa$. 
To achieve a fixed statistical accuracy, estimating the expectation value within the error $\epsilon$ requires an additional $\mathcal{O}(\kappa^2/\epsilon^2)$ shots~\cite{Temme2017}.
This is commonly referred to as \textit{sampling overhead} quantified by $\kappa$.
\looseness=-1

Minimizing the sampling overhead is highly desirable, and thus, the search for a decomposition that minimizes the value of $\kappa$ is of great importance. 
The optimal sampling overhead with respect to the set $\locc(A,B)$ for a given operator $\mathcal{E}$ is referred to as $\gamma(\mathcal{E})$~\cite{Piveteau2022}.
It is given by
\begin{equation}\label{eq:min_sampling_overhead}
    \begin{split}
        \gamma(\mathcal{E}) := \min\biggl\{\sum_i |c_i| \,\bigg|\,&  \mathcal{E} = \sum_i c_i \mathcal{F}_i, c_i\in\mathbb{R}, \\
        &\mathcal{F}_i \in \locc(A,B)\biggr\}.
    \end{split}
\end{equation}

\subsection{Quasiprobability simulation of maximally entangled states}\label{sec:quasi_prob_sim_states}
While the previously introduced quasiprobability simulation of non-local transformations aims to produce a superoperator, simulating an entangled state involves generating a density operator. 
However, the quasiprobability simulation of a maximally entangled state $\Phi_{AB}$ from a bipartite NME state $\rho_{AB}$ can be reduced to simulating a superoperator $\mathcal{E}\in \cptn(A\otimes B)$ that transforms $\rho_{AB}$ into $\Phi_{AB}$, i.e., $\mathcal{E}(\rho_{AB}) = \Phi_{AB}$~\cite{Piveteau2022}.
Given that  $\mathcal{E}$ converts an NME state $\rho_{AB}$ into a maximally entangled one, it cannot be from $\locc(A,B)$~\cite{Chitambar2014}.

Nevertheless, a QPD as formulated in \Cref{eq:QPD} can be utilized to decompose $\mathcal{E}$ using operators within $\locc(A,B)$.
Therefore, a QPD of the maximally entangled state $\Phi_{AB}$ utilizing the state $\rho_{AB}$ is given by
\begin{align}
    \Phi_{AB} &= \mathcal{E}(\rho_{AB})\\
    &= \sum_{i=1}^m c_i \mathcal{F}_i(\rho_{AB}) \label{eq:QPD_state}
\end{align}
where $\mathcal{F}_i \in \locc(A,B)$. 
Using this QPD, the Monte Carlo approach from \Cref{eq:expectation} can be applied to compute the expectation value for an observable $O$.

This allows us to define the optimal sampling overhead for the quasiprobability simulation for the maximally entangled state $\Phi_{AB}$ given $\rho_{AB}$ as a resource as
\begin{align}
    \hat{\gamma}^{\rho_{AB}}(\Phi_{AB}) := \min\left\{ \gamma(\mathcal{E}) \,\middle|\,   \mathcal{E}(\rho_{AB}) = \Phi_{AB} \right\}
\end{align}
which, according to~\cite[Proposition 13]{Yuan2023}, is given by
\begin{align}\label{eq:optimal_overhead}
    \hat{\gamma}^{\rho_{AB}}(\Phi_{AB}) = \frac{2}{f(\rho_{AB})} - 1,
\end{align}
where $f(\rho_{AB})$ is the maximum overlap as defined in \Cref{eq:overlap}. 
Hence, the greater the degree of entanglement in the state $\rho_{AB}$, as measured by $f(\rho_{AB})$, the less sampling overhead is necessary for simulating a maximally entangled state.

\subsection{Wire cutting as quasiprobability simulation}
\begin{figure}
    \centering
    \input{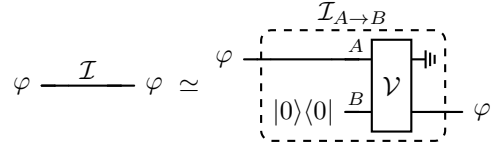}
    \vspace{-0.2cm}
    \caption{Identity $\mathcal{I}_{A\rightarrow B}$ between systems $A$ and $B$ modeled by a non-local operation $\mathcal{V}$ where $A$ is traced out and discarded, symbolized by $\begin{quantikz}[column sep=4pt]\qw&\ground{}\end{quantikz}$~\cite{Brenner2023}.}
    \label{fig:wire_cut_model}
\end{figure}
\begin{figure*}
    \centering
    \input{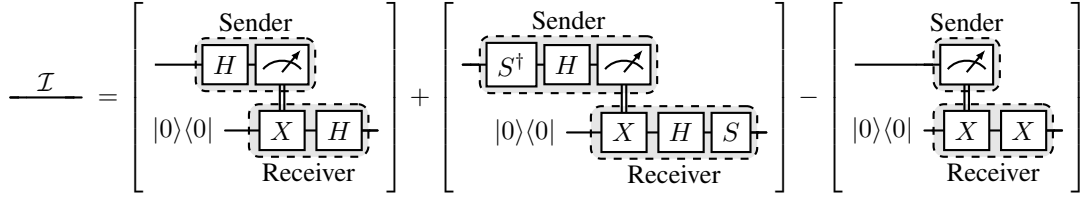}
    \vspace{-0.2cm}
    \caption{Optimal wire cut for a single wire~\cite{Harada2023}.}
    \vspace{-0.4cm}
    \label{fig:wire_cut_harada}
\end{figure*}
Initially proposed by Peng et al.~\cite{Peng2019}, wire cutting quasiprobabilistically simulates the transmission of a quantum state.
It enables simulating the transfer of a state from a qubit of Hilbert space $A$ on one quantum device to another qubit of Hilbert space $B$ on a different device. 
This allows the execution of a circuit to be split up between quantum devices by cutting wires.
To apply a quasiprobability simulation, as introduced in \Cref{sec:quasi_prob_sim}, to a wire, the wire must be represented as a non-local operator. 
This is depicted in \Cref{fig:wire_cut_model}, where an empty wire of a circuit on the left, represented by the identity operation $\mathcal{I}$, is modeled as a non-local operator $\mathcal{V}$ connecting $A$ and $B$ on the right.
This operator $\mathcal{V}$ acts as the non-local identity $\mathcal{I}_{A \rightarrow B}$ between qubits $A$ and $B$, when system $A$ is traced out~\cite{Brenner2023}: 
\begin{align}
\forall \varphi \in D(A): \operatorname{Tr}_{A}[\mathcal{V}(\varphi \otimes \ketbra{0}{0}_B)] = \varphi,
\end{align}
with $\ketbra{0}{0}_B$ being the initial state of qubit $B$.
By using a QPD for the non-local operator $\mathcal{V}$, i.e. $\mathcal{V}=\sum_i c_i \mathcal{F}_i$ with $\mathcal{F}_i \in \locc(A,B)$, wire cutting can be understood as the quasiprobabilistic simulation of the identity operator $\mathcal{I}_{A \rightarrow B}$: 
\begin{align}
\forall \varphi \in D(A): \sum_i c_i \operatorname{Tr}_{A}[\mathcal{F}_i(\varphi \otimes \ketbra{0}{0}_B)] = \varphi.
\end{align}

The minimal sampling overhead for cutting a single wire with local operations and classical communication has been identified as $\gamma(\mathcal{I}) = 3$~\cite{Brenner2023}. 
Harada et al.~\cite{Harada2023} demonstrated a representation of a wire cut with minimal sampling overhead using the following QPD of the one-qubit identity operator $\mathcal{I}$:
\begin{equation}\label{eq:wire_cut_harada}
    \begin{split}
        \mathcal{I}(\bullet)
    =&\sum_{i\in\{1,2\}} \sum_{j\in\{0,1\}} \mathrm{Tr}\left[ U_{i} \ketbras{j} U_{i}^{\dagger}(\bullet)_A \right]U_{i} \ketbras{j}_B U_{i}^{\dagger}\!\! \\
    &- \sum_{j\in\{0,1\}} \mathrm{Tr}\left[ \ketbras{j}(\bullet)_A \right]X \ketbras{j}_BX,
    \end{split}
\end{equation}
where $U_1=H$ and $U_2=SH$, with $H$ the Hadamard gate, $S$ the phase gate, and $X$ the Pauli $X$ gate.

The corresponding wire cut circuits are shown in \Cref{fig:wire_cut_harada}, with measurements performed on qubit $A$ at the sender's device and the corresponding states initialized on qubit $B$ at the receiver's side.
Classical communication between the devices is facilitated by a classical controlled-not gate, represented by a double line connecting the measurement and the $X$ gate.

\begin{figure}[t]
    \centering
    \input{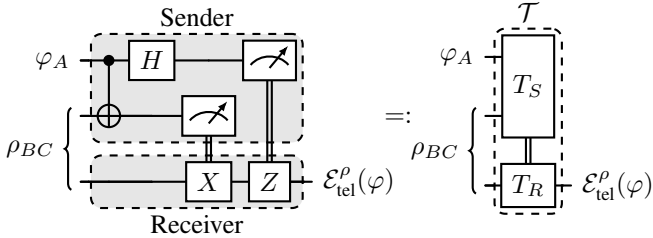}
    \vspace{-0.2cm}
    \caption{Quantum teleportation circuit using resource state $\rho$ with detailed sender and receiver operations (left) and their condensed notation (right).}
    \label{fig:teleportation_circuit}
\end{figure}

\subsection{Quantum Teleportation}
As opposed to quasiprobabilistically simulating the transmission of a qubit's state, quantum teleportation actually transmits the state using entanglement~\cite{Bennett1993}.
This protocol uses local operators only and involves the communication of only two classical bits without physically transmitting the state itself.
To ensure that the state is correctly transmitted, quantum teleportation requires a pair of maximally entangled qubits $\Phi$~\cite{Prakash2012} that are shared between sender and receiver.

The teleportation circuit consists of three qubits $A$, $B$, and $C$ and is depicted on the left side of \Cref{fig:teleportation_circuit} with an entangled resource state $\rho_{BC}$.
The sender performs a Bell basis measurement on the qubit $A$ containing the state $\varphi$ to be teleported and the qubit $B$ of the resource state. 
The measurement outcome is then transmitted to the receiver via a classical channel. 
Using this information, the receiver applies appropriate transformations to their part of the entangled pair, i.e. the qubit $C$, to reconstruct the original state $\varphi$ at the receiver's location.
The right side of \Cref{fig:teleportation_circuit} shows a condensed representation of the teleportation circuit, encapsulating the sender's and receiver's operations into the transformations $T_S$ and $T_R$, respectively.
The entire protocol is described as the operator $\mathcal{T}$, which acts on the state $\varphi_{A}$ and the resource $\rho_{BC}$, expressed as $\mathcal{T}(\varphi_{A} \otimes \rho_{BC}).$
\looseness=-1

Utilizing quantum teleportation with an arbitrary NME resource state $\rho_{BC}$, which may not be maximally entangled, results in the state given by~\cite{Gu2004}
\begin{align}\label{eq:tele}
    \mathcal{E}_{\text{tel}}^{\rho}(\varphi) &= \operatorname{Tr}_{AB}\left[\mathcal{T}(\varphi_A \otimes \rho_{BC})\right] \\
    &=\sum_{\sigma \in \{I, X, Y, Z\}}\braket{\Phi^{\sigma}|\rho|\Phi^{\sigma}}\sigma\varphi\sigma.\label{eq:tele2}
\end{align}
Here, $\ket{\Phi^{\sigma}}$ denotes the Bell basis states associated with the Pauli operators  $\sigma \in \{I, X, Y, Z\}$. 
These Bell states are defined as $\ket{\Phi^{\sigma}} = (\sigma \otimes I)\ket{\Phi}$.
Consequently, depending on the overlap $\braket{\Phi^{\sigma}|\rho|\Phi^{\sigma}}$ of the resource state $\rho$ with the different Bell states, corresponding Pauli errors are introduced during teleportation.

\begin{figure}
    \centering
    \input{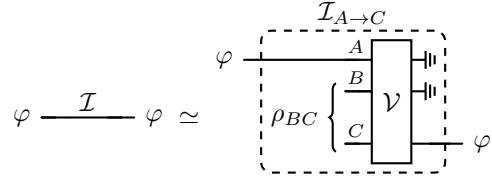}
    \vspace{-0.2cm}
    \caption{A wire modeled as a non-local three-qubit operation $\mathcal{V}$ on qubits $A$, $B$, and $C$, employing the NME state $\rho$.}
    \label{fig:wire_cut_model_ent}
\end{figure}
\begin{figure*}
    \centering
    \input{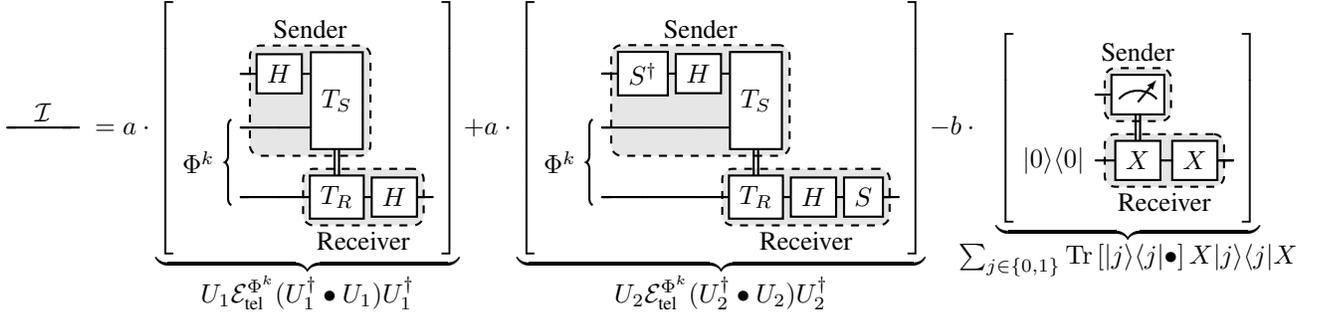}
    \vspace{-0.2cm}
    \caption{Wire cut that uses NME states $\ket{\Phi^k}$ with $a=\frac{k^2+1}{(k+1)^2}$ and $b=\frac{(k-1)^2}{(k+1)^2}$, employing the teleportation circuit from \Cref{fig:teleportation_circuit}.}
    \label{fig:cut_with_entanglement}
\end{figure*}  

\section{Optimal wire cutting with NME states}\label{sec:main-section}

In this section, we propose an enhancement to the wire cutting procedure by integrating NME qubit pairs, aiming to reduce the sampling overhead based on their degree of entanglement. 
Therefore, we first outline how wire cutting with NME states can be modeled and define its optimal sampling overhead. 
Subsequently, we will determine the optimal sampling overhead, as articulated in \Cref{theorem_overhead}. 
Following this, \Cref{theorem_decomposition} will present a QPD to realize this optimal sampling overhead with pure NME states.

To define a wire cut using an entangled state, we adapt the model as a non-local operator as described in \Cref{fig:wire_cut_model}~\cite{Brenner2023}.
Therefore, we represent the identity operator $\mathcal{I}$ as some non-local operator $\mathcal{V}$ acting on three qubits denoted by their Hilbert spaces $A$, $B$, and $C$, as depicted in \Cref{fig:wire_cut_model_ent}.
This allows the sharing of an NME state $\rho_{BC}$ between a sender system $A\otimes B$ and a receiver system $C$.
Similar to wire cutting without entangled states, operator $\mathcal{V}$ mimics the identity $\mathcal{I}_{A \rightarrow C}$ between $A$ and $C$ when system $A \otimes B$ is traced out:
\begin{align}
        \forall \varphi \in D(A):\operatorname{Tr}_{AB}[\mathcal{V}(\varphi\otimes \rho_{BC})] = \varphi.
\end{align}

In analogy to \Cref{eq:min_sampling_overhead}, the optimal sampling overhead for wire cutting with NME states $\rho$ is denoted as $\gamma^{\rho}(\mathcal{I})$, where the superscript $\rho$ distinguishes it from the sampling overhead without entangled states.
It is defined as:
\begin{equation}\label{eq:min_sampling_overhead_wire_nme}
    \begin{split}
        \gamma^{\rho}(\mathcal{I}) := \min\biggl\{\sum_i |c_i|  \,\bigg|\,&  \mathcal{V} = \sum_i c_i \mathcal{F}_i, c_i\in\mathbb{R},\\ &\mathcal{F}_i \in \locc(A\otimes B, C), \\
        &\forall \varphi \in D(A):\\
        &\quad\operatorname{Tr}_{AB}[\mathcal{V}(\varphi\otimes \rho_{BC})] = \varphi\biggr\}.
    \end{split}
\end{equation}
Based on this definition, the following theorem quantifies the sampling overhead $\gamma^{\rho}(\mathcal{I})$. 

\begin{theorem} \label{theorem_overhead}
The optimal sampling overhead for a wire cut of a single-qubit identity operator $\mathcal{I}$ using an arbitrary two-qubit NME resource state described by density operator $\rho$ is  
\begin{equation}
    \gamma^{\rho}(\mathcal{I}) = \frac{2}{f(\rho)} -1.
\end{equation}
\end{theorem}

The detailed proof that reduces the sampling overhead $\gamma^{\rho}(\mathcal{I})$ of the wire cut to the sampling overhead $\hat{\gamma}^{\rho}(\Phi)$ of simulating the maximally entangled state using $\rho$ is provided in Appendix~\ref{sec:proof_overhead}.
Consequently, the higher the degree of entanglement in the resource state $\rho$ provided in the wire cut, the lower the resulting sampling overhead $\gamma^{\rho}(\mathcal{I})$.
In the absence of entanglement, i.e., $f(\rho) = \frac{1}{2}$, the sampling overhead aligns with the previously established optimal value for the one-qubit wire cut, i.e., $\gamma(\mathcal{I})=3$~\cite{Brenner2023}.
Conversely, the sampling overhead is eliminated for maximally entangled states, where $f(\rho) = 1$.
This is attributed to the feasibility of employing standard quantum teleportation under these conditions~\cite{Bennett1993}.
As a result, the following relationship is established for the use of pure NME states $\Phi^k$:
\begin{corollary}\label{corollary_pure_overhead}
The optimal sampling overhead for a wire cut of a single-qubit identity operator $\mathcal{I}$ using pure NME resource states $\Phi^k$ is given by 
\begin{equation}
    \gamma^{\Phi^k}(\mathcal{I}) = \frac{4(k^2+1)}{(k+1)^{2}} -1.
\end{equation}
\end{corollary}
\begin{proof}
This follows from \Cref{theorem_overhead} and \Cref{eq:singlet_fraction_nme}.
\end{proof}

To achieve this optimal sampling overhead given pure NME states $\Phi^k$, the following theorem provides the optimal QPD.

\begin{theorem} \label{theorem_decomposition}
The single-qubit identity operator $\mathcal{I}$ can be decomposed using quantum teleportation described by $\mathcal{E}_{\text{tel}}^{\Phi^k}$ with pure NME resource states $\Phi^k$ as 
\begin{equation}\label{eq:thrm2}
    \begin{split}
    \mathcal{I}(\bullet) =&\phantom{-} \frac{k^2+1}{(k+1)^2}\sum_{i\in\{1,2\}} U_i\mathcal{E}_{\text{tel}}^{\Phi^k}(U_i^{\dagger}(\bullet) U_i)U_i^{\dagger}\\
    &- \frac{(k-1)^2}{(k+1)^2}\sum_{j\in\{0,1\}} \tr\left[ \ketbras{j}(\bullet) \right]X \ketbras{j}X,
    \end{split}
\end{equation}
where $U_1=H$ and $U_2=SH$. 
This decomposition achieves the optimal sampling overhead from \Cref{corollary_pure_overhead}.
\end{theorem}

The proof of \Cref{theorem_decomposition} is given in Appendix~\ref{sec:proof_decomposition}.
The quantum circuits in \Cref{fig:cut_with_entanglement} illustrate this theorem's application. 
They contain two teleportation circuits where the transmitted state undergoes a transformation involving $H$ and $S$ before teleportation, followed by the inverse transformation after the teleportation.
Each teleportation process uses an instance of NME state $\ket{\Phi^k}$.
The decomposition also includes a circuit that measures the qubit state and then initializes a state corresponding to the inverse of the measured result.
This circuit does not consume an entangled state.
The decomposition in \Cref{theorem_decomposition} generalizes the result from Harada et al.~\cite{Harada2023} for a single-qubit wire cut as given in \Cref{eq:wire_cut_harada}.
The measurement and subsequent initializations in the first two circuits of the wire cut in \Cref{fig:wire_cut_harada} are replaced by a quantum teleportation using $\ket{\Phi^k}$.
As a result, the coefficients in the QPD decrease when the entanglement in the resource state $\ket{\Phi^k}$ increases.
Correspondingly, this results in a reduction of the associated sampling overhead.

When sampling from the QPD of \Cref{theorem_decomposition}, the required number of entangled states $\ket{\Phi^k}$ corresponds to the total number of teleportations performed. 
This number is proportional to $\frac{2(k^2+1)}{(k+1)^2} = \bra{\Phi}\Phi^k\ket{\Phi}^{-1}$. 
As a result, an increase in the entanglement level represented by values of $k$ closer to $1$ leads to a decrease in the number of entangled states needed to achieve the desired level of statistical accuracy.

\section{Numerical Experiments}\label{sec:num_experiments}

This section presents a numerical demonstration using a simulator that shows the advantages of incorporating NME states in wire cutting.
The code and all data generated are publicly available at~\cite{codeanddata}.
We evaluate the error in the expectation values resulting from applying the wire cut of \Cref{theorem_decomposition} to randomly sampled quantum states.
Specifically, we measure the expectation value of Pauli $Z$.
The relationship between the error and the total number of executed shots shows the sampling overhead necessary for replicating the expectation value with wire cutting and achieving a certain accuracy.

For each random state, the cutting procedure is conducted with NME states according to \Cref{eq:nme} with varying degrees of entanglement measured by their overlap with the maximally entangled state as introduced in \Cref{eq:overlap}.
We follow a specific procedure to initialize a qubit in a random state: 
A unitary matrix $W$ is randomly sampled~\cite{Mezzadri2007} and applied to the initial state $\ket{0}$. 
The resulting state is given by $W\!\ket{0}$.
To determine the error, the exact expectation value is computed classically by calculating $\braket{Z}_{W\!\ket{0}} = \braket{0|W^{\dagger}ZW|0}$.

Next, the same expectation value with wire cutting is computed by sampling from the simulator.
The wire cut is applied to the state $W\!\ket{0}$, resulting in three subcircuits as illustrated in \Cref{fig:cut_with_entanglement}.
The expectation value of each subcircuit is measured, and the results are combined following \Cref{theorem_decomposition}. 
To execute these subcircuits, the Qiskit Aer simulator is utilized~\cite{Qiskit}. 
A fixed number of shots is allocated collectively to all three subcircuits, distributed proportionally to their coefficients as defined in \Cref{theorem_decomposition}.
The error is the absolute deviation between the sampled expectation value $\braket{Z}_{W\!\ket{0}}^{\text{sample}}$ with wire cutting and the exact expectation value $\braket{Z}_{W\!\ket{0}}$:
\begin{align}
    \epsilon:=\left|\braket{Z}_{W\!\ket{0}}^{\text{sample}} - \braket{Z}_{W\!\ket{0}}\right|.
\end{align}
We vary the total number of shots in the experiments between 0 and 5000 to evaluate the error depending on the number of shots.
Additionally, this experiment is performed for varying amounts of entanglement $f(\Phi^k) \in \{0.5, 0.6, 0.7, 0.8, 0.9, 1.0\}$ in the pre-shared qubit pair $\ket{\Phi^k}$.
Furthermore, we consider 1000 random states $W\!\ket{0}$ as input for this evaluation and compute the average error over all inputs.

\begin{figure}[t]
    \centering
    \includegraphics{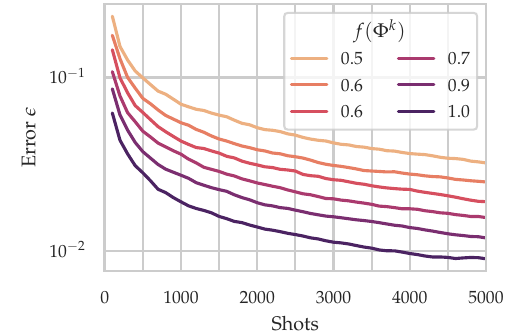}
    \caption{Average error in the expectation value when applying wire cutting with different degrees of entanglement.}
    \label{fig:experiment_plot} 
\end{figure}

The results are presented in \Cref{fig:experiment_plot}, where the error is plotted on a logarithmic scale.
They demonstrate a clear relationship between the degree of entanglement in the resource states, quantified by $f(\Phi^k)$, and the error showing the deviation from the exact result.
Higher entanglement degrees correspond to lower errors for a fixed number of shots. 
The case where $f(\Phi^k)=1$ serves as a baseline, employing quantum teleportation with maximally entangled states.
In this scenario, there is no sampling overhead from the QPD, and the error arises only from statistical errors caused by the finite number of shots since the exact state is teleported.
Conversely, when there is no entanglement $f(\Phi^k)=0.5$, default wire cutting is performed, leading to the largest errors.
These observations align with the theoretical considerations regarding the sampling overhead of wire cutting with NME states in \Cref{sec:main-section}.

\section{Related Work}\label{seq:related_work}

The quasiprobability simulation of maximally entangled states using NME states, as outlined in \Cref{sec:quasi_prob_sim_states}, is explored extensively in the context of virtual entanglement distillation \cite{Yuan2023}.
In contrast to traditional entanglement distillation, which tries to convert an NME state into a physical state of higher entanglement~\cite{Bennett1996a,Bennett1996b,Rozpedek2018}, virtual distillation replicates only its measurement statistics via quasiprobability simulation.
This method has recently been tested experimentally~\cite{Zhang2023b}.
The resulting virtual entangled states can be utilized in various quantum protocols, such as teleportation.  
Consequently, employing these virtual states in quantum teleportation enables the implementation of a wire cut.
However, our research takes a different direction, focusing on the direct implementation of a wire cut with NME states, which can be seen as a virtual distillation of the identity operator~\cite{Yuan2023}.
This direct approach eliminates the need to separately simulate a maximally entangled state before its application in teleportation.

Moreover, an alternative circuit cutting technique to wire cutting is gate cutting~\cite{Mitarai2021,Piveteau2022,Ufrecht2023a,Ufrecht2023b}. 
It involves decomposing multi-qubit gates to cut a circuit. 
Therein, a quasiprobability simulation of the corresponding non-local operator of the gate is applied as described in \Cref{sec:quasi_prob_sim}.
Depending on the characteristics of the circuit, either a wire cut or gate cut can be more favorable to achieve a minimal sampling overhead~\cite{Brenner2023}. 

Additionally, Brenner et al.~\cite{Brenner2023} show that cutting multiple wires together can effectively reduce the sampling overhead, compared to cutting each one individually.
Methods for automatically finding optimal positions for wire and gate cuts have been developed~\cite{Tang2021,Brandhofer2023a}.
Moreover, to address finite-shot error, maximum likelihood methods have been introduced in wire cutting~\cite{Perlin2021}.
Although circuit cutting primarily aims to reduce circuit size, some studies empirically demonstrate that executing smaller subcircuits can enhance overall results~\cite{Ayral2021,Bechtold2023,Perlin2021}.
\looseness=-1

Additionally, previous research has investigated using NME states as resource states in quantum teleportation. 
When employing the original protocol with NME states, the fidelity of the teleported state is reduced~\cite{Prakash2012}. 
Unit fidelity can only be achieved when using maximally entangled states. 
However, the modified probabilistic teleportation protocol allows the teleportation of an unknown state with NME resource states while maintaining maximal fidelity~\cite{Agrawal2002,Pati2004}. 
Nevertheless, this protocol can fail, so it must be repeated if attempts are unsuccessful. 
This repetition introduces an overhead to achieve the desired number of successful teleportations.

\section{Conclusion}\label{sec:conclusion}
This work demonstrates that wire cutting can leverage NME states to lower its associated cost in terms of the sampling overhead when compared to wire cutting without entanglement. 
Increasing the degree of entanglement in the resource states can effectively reduce the sampling overhead. 
The optimal sampling overhead for arbitrary NME states is derived in \Cref{theorem_overhead}.
This result highlights the significance of entanglement as a valuable computational resource. 
Moreover, \Cref{theorem_decomposition} presents a wire cut employing pure NME states that achieves the optimal sampling overhead. 
With the findings presented in this paper, we can now bridge the gap between circuit cutting and quantum teleportation using NME resource states.
Consequently, this advancement significantly increases the flexibility in implementing distributed quantum computing.

Future work can explore wire cutting protocols using mixed NME states, considering the resilience and performance of these protocols in the presence of noise inherent in contemporary quantum devices. 
Additionally, using NME states for multiple wire cuts in parallel may lead to further reductions in sampling overhead, as observed in wire cutting without NME states \cite{Brenner2023}.
Another interesting research avenue is how the results of using NME states in wire cutting transfer to gate cutting techniques to increase the efficiency of decomposing multi-qubit gates.

\bibliographystyle{IEEEtran}
\bibliography{bibliography}

\appendices
\section{Maximal overlap for $\ket{\Phi^k}$}\label{sec:singlet_fraction}

For pure states $\ket{\psi}_{AB}$, there are existing results~\cite[Theorem 15]{Regula2019} that show that the overlap $f(\psi_{AB})$ from \Cref{eq:overlap} relates to the $m$-distillation norm with $m=2$ as
\begin{align}\label{eq:pure_overlap}
    f(\psi_{AB}) = \frac{1}{2}\|\ket{\psi}_{AB}\|^2_{[2]}.
\end{align}
The $m$-distillation norm is defined as~\cite[Theorem 2]{Regula2018}
\begin{align}\label{eq:m_distill_norm}
    \|\ket{\psi}_{AB}\|_{[m]} := \|\zeta^{\downarrow}_{1:j^*}\|_{1} + \sqrt{j^*}\|\zeta^{\downarrow}_{j^*+1:d}\|_{2}
\end{align}
with $d = \min(\dim(A),\dim(B))$, and $\|\bullet\|_{1}$, $\|\bullet\|_{2}$ representing the $1$-norm and $2$-norm of vectors, respectively.
Moreover, $\zeta^{\downarrow}_{1:j}$ represents the vector of the $j$ largest Schmidt coefficients of $\ket{\psi}_{AB}$, and $\zeta^{\downarrow}_{j+1:d}$ contains the $d-j$ smallest Schmidt coefficients. 
The vector $\zeta^{\downarrow}_{a:b}$ is defined as the zero vector for $a>b$. 
The value $j^*$ is defined as
\begin{align}\label{eq:k_star}
j^* := \argmin_{1\le j \le m}\frac{1}{j}\| \zeta^{\downarrow}_{m-j+1:d} \|_{2}^2.
\end{align}

To calculate $f(\Phi^k)$ with \Cref{eq:pure_overlap}, the $2$-distillation norm of $\ket{\Phi^k}$ must be computed.
Since $m= 2$, \Cref{eq:k_star} allows two possible values of $j^{*}$, nameley $j^{*} \in \{1, 2\}$.
However, since $\ket{\Phi^k}$ has only two non-zero Schmidt coefficients, $K$ and $kK$ with $K=\sqrt{1+k^2}^{-1}$, both values $j^{*}=1$ and $j^{*}=2$ yield the same result in \Cref{eq:m_distill_norm}: 
\begin{align}
    \text{for } j^{*}=1 :\|\zeta^{\downarrow}_{1:1}\|_{1} + \sqrt{1}\|\zeta^{\downarrow}_{1+1:d}\|_{2} &= \|\zeta^{\downarrow}_{1:2}\|_{1}, \label{eq:mdist1}\\
    \text{for } j^{*}=2 :\|\zeta^{\downarrow}_{1:2}\|_{1} + \sqrt{2}\|\zeta^{\downarrow}_{2+1:d}\|_{2} &= \|\zeta^{\downarrow}_{1:2}\|_{1}\label{eq:mdist2}.
\end{align}
Herein, \Cref{eq:mdist1} holds since the second summand is the 2-norm over a one-element vector, which results in the absolute value of its element, and \Cref{eq:mdist2} holds since the second summand is the 2-norm over the zero vector.

Consequently, by using \Cref{eq:m_distill_norm} with \Cref{eq:mdist1,eq:mdist2} the $2$-distillation norm of $\ket{\Phi^k}$ is determined by the $1$-norm of its Schmidt coefficients:
\begin{align}
    \|\ket{\Phi^k}\|_{[2]} &= \|\zeta^{\downarrow}_{1:2}\|_{1} \\
    &= \|(K, kK)^T\|_{1} \\
    &= |K| + |kK| \\
    &= K(1+k)  \label{eq:positive}
\end{align}
where the last equality holds since $k\in\mathbb{R}_{\ge 0}$.
Therefore, by using \Cref{eq:pure_overlap}:
\begin{align}
    f(\Phi^k) &= \frac{1}{2}\|\ket{\Phi^k}\|^2_{[2]} \\
    &= \frac{1}{2} (K(k+1))^2  \label{eq:overlap_fidelity}\\
    &= \frac{(k+1)^2}{2(k^2 + 1)}.
\end{align}

Furthermore, for the pure state $\Phi^k$, it holds that
\begin{align}
    \braket{\Phi|\Phi^k|\Phi} &= \left|\braket{\Phi|\Phi^k}\right|^2\\
    &= \left| \frac{1}{\sqrt{2}}K + \frac{1}{\sqrt{2}}kK \right|^2\\
    &= \frac{1}{2} (K(k+1))^2,
\end{align}
which, using \Cref{eq:overlap_fidelity}, shows that the maximal overlap matches the overlap with $\Phi^k$ itself: $f(\Phi^k) = \braket{\Phi|\Phi^k|\Phi}$.

\section{Proof of \Cref{theorem_overhead}}\label{sec:proof_overhead}

\begin{proof}
    We aim to prove that the sampling overhead $\gamma^{\rho}(\mathcal{I})$ for cutting a single wire using arbitrary two-qubit states $\rho$ is equal to the sampling overhead $\hat{\gamma}^{\rho}(\Phi)$ for simulating the maximally entangled state $\Phi$ using $\rho$. 
    This is demonstrated by showing that $\hat{\gamma}^{\rho}(\Phi)$ forms both an upper and lower bound for $\gamma^{\rho}(\mathcal{I})$.
    
    To show the upper bound, i.e.,  $\gamma^{\rho}(\mathcal{I}) \le \hat{\gamma}^{\rho}(\Phi)$, we construct a wire cut based on a QPD for the maximally entangled state $\Phi$ using the NME state $\rho$ by employing quantum teleportation.
    This teleportation-based realization of a wire cut was used to analyze wire cutting without NME states~\cite{Brenner2023}.
    To incorporate NME states in this construction, consider three qubits, each represented by their respective Hilbert spaces  $A$, $B$, and $C$. 
    Moreover, consider a QPD for the maximally entangled state $\Phi_{BC}$ using $\rho_{BC}$ with an optimal sampling overhead $\hat{\gamma}^{\rho}(\Phi)$ as given in \Cref{eq:optimal_overhead}.
    This QPD is expressed as $\Phi_{BC} = \sum_i c_i \mathcal{F}_i(\rho_{BC})$ with $\mathcal{F}_i \in \locc(B,C)$.
    Applying this QPD within the quantum teleportation protocol governed by operator $\mathcal{T}$ forms a QPD for the identity operator $\mathcal{I}_{A \rightarrow C}$ between qubit $A$ and $C$ when the system $A\otimes B$ is traced out:
    \begin{align}
        \forall \varphi \in D(A): \varphi &= \operatorname{Tr}_{AB}\left[\mathcal{T}(\varphi\otimes\Phi_{BC})\right]\\
        &= \sum_i c_i \operatorname{Tr}_{AB}\left[\mathcal{T}(\varphi \otimes \mathcal{F}_i(\rho_{BC}))\right].
    \end{align}
    Each operator within the QPD involves only local operations and classical communication between $A\otimes B$ and $C$: $\mathcal{T}\in \locc(A\otimes B, C)$ and $\mathcal{F}_i \in \locc(B,C)$. 
    Therefore, this creates a wire cut using NME states $\rho_{BC}$.
    Since the teleportation $\mathcal{T}$ uses the simulated maximally entangled state $\Phi_{BC}$, it does not introduce additional sampling overhead.
    Thus, the sampling overhead for this protocol is determined by the sampling overhead for the QPD for $\Phi_{BC}$.
    Consequently, this teleportation-based realization of a wire cut with a simulated maximally entangled state establishes the upper bound of the optimal sampling overhead using $\rho$, i.e. $\gamma^{\rho}(\mathcal{I}) \le \hat{\gamma}^{\rho}(\Phi)$.

    To establish the lower bound, namely $\gamma^{\rho}(\mathcal{I}) \ge \hat{\gamma}^{\rho}(\Phi)$, we adapt the proof methodology from Brenner et al.~\cite{Brenner2023}, which was originally applied to wire cutting without entangled states. 
    Our adaptation involves employing the newly introduced model of wire cutting using NME states, conceptualized as a non-local three-qubit operator $\mathcal{V}$, as depicted in \Cref{fig:wire_cut_model_ent}. 
    We show that a QPD for $\mathcal{V}$ with an optimal sampling overhead $\gamma^{\rho}(\mathcal{I})$ can be used to construct a QPD for the maximally entangled state $\Phi$ using NME state $\rho$ with the same overhead.

    \begin{figure}
        \centering
        \input{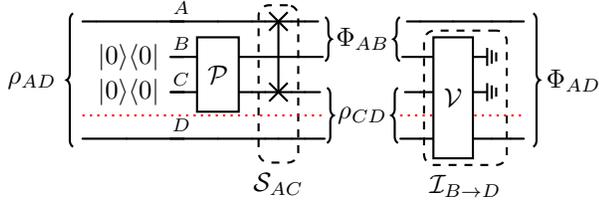}
        \vspace{-0.2cm}
        \caption{Circuit used in the proof, with qubits separated into subsystems $A\otimes B \otimes C$ and $D$ by the red dotted line.}
        \label{fig:proof_circuit}
    \end{figure}

    Consider a four-qubit system described by their respective Hilbert spaces $A$, $B$, $C$, and $D$ partitioned into two subsystems $A\otimes B \otimes C$ and $D$.
    Within the composed system, $A\otimes D$ initially holds an NME state $\rho_{AD}$, while qubits $B$ and $C$ are both in the initial state $\ketbras{0}$ as depicted on the left in \Cref{fig:proof_circuit}.
    Since the only shared state between the two subsystems is $\rho_{AD}$, the bipartite entanglement is solely characterized by this state.

    This initial state is transformed by first applying an operator $\mathcal{P}(\ketbras{00}_{BC}) = \Phi_{BC}$ that prepares the maximally entangled state $\Phi_{BC}$ on qubits $B$ and $C$. 
    Then, qubits $A$ and $C$ exchange states via the application of the swap operator $\mathcal{S}_{AC}$.
    The resulting state, depicted in the middle of \Cref{fig:proof_circuit}, is
    \begin{align*}
        \mathcal{S}_{AC}(\mathcal{P}(\ketbras{00}_{BC}) \otimes \rho_{AD}) = \Phi_{AB}\otimes \rho_{CD}.
    \end{align*}
    As $\mathcal{P}$ and $\mathcal{S}_{AC}$ solely operate on the subsystem $A\otimes B \otimes C$, they qualify as local operators with respect to the subsystems.

    To produce the maximally entangled state $\Phi_{AD}$ using $\rho_{CD}$ after the above state preparation, we implement the identity $\mathcal{I}_{B \rightarrow D}$. 
    This involves transferring the state of qubit $B$ to $D$ via the non-local operator $\mathcal{V}$, as depicted on the right of \Cref{fig:proof_circuit}. 
    To simulate $\mathcal{V}$ with local transformations, consider the QPD achieving the optimal sampling overhead $\gamma^{\rho}(\mathcal{I})$ using $\rho_{CD}$:
    \begin{align}
        \mathcal{V} = \sum_i c_i \mathcal{F}_i,
    \end{align} 
    where $\mathcal{F}_i \in \locc(B\otimes C,D)$.
    The sampling overhead for this QPD is given as $\sum_i |c_i| = \gamma^{\rho}(\mathcal{I})$, as defined by \Cref{eq:min_sampling_overhead_wire_nme}.

    By using this QPD, we can facilitate the quasiprobabilistic simulation of a maximally entangled qubit pair $\Phi_{AD}$ between the two subsystems $A\otimes B \otimes C$ and $D$  from $\Phi_{AB}\otimes \rho_{CD}$:
    \begin{align}
        \Phi_{AD} &= (\mathcal{I}_A \otimes \mathcal{I}_{B \rightarrow D})(\Phi_{AB}\otimes \rho_{CD})\\
        &= \operatorname{Tr}_{BC}[(\mathcal{I}_A \otimes \mathcal{V})(\Phi_{AB}\otimes \rho_{CD})] \\
        &= \sum_i c_i \operatorname{Tr}_{BC}[(\mathcal{I}_A \otimes \mathcal{F}_i)(\Phi_{AB}\otimes \rho_{CD})].\label{eq:phi_ab}
    \end{align}
    As a result, for the operators $\widetilde{\mathcal{F}}_i$ defined as 
    \begin{align}
    \begin{split}
        &\widetilde{\mathcal{F}}_i(\ketbras{00}_{BC} \otimes \rho_{AD}) \\
        &\quad:= (\mathcal{I_A} \otimes \mathcal{F}_i)\mathcal{S}_{AC}(\mathcal{P}(\ketbras{00}_{BC}) \otimes \rho_{AD})
        \end{split}\\
        &\quad= (\mathcal{I_A} \otimes \mathcal{F}_i)(\Phi_{AB}\otimes \rho_{CD}),
    \end{align}
    it holds that $\widetilde{\mathcal{F}}_i \in \locc(A\otimes B \otimes C, D)$.

    Given that appending qubits in an initial state $\ketbras{00}_{BC}$ and subsequently tracing them out are both completely positive and trace-preserving operators that take place in the same subsystem, the operators defined as
    \begin{align}
        \mathcal{G}_i(\rho_{AD}) =  \operatorname{Tr}_{BC} \left[\widetilde{\mathcal{F}}_i(\ketbras{00}_{BC} \otimes \rho_{AD}))\right],
    \end{align}
    are elements of $\locc(A,D)$.

    As a result, we obtain a QPD for the maximally entangled state between $A$ and $D$:
    \begin{align}
        \Phi_{AD} &= \sum_i c_i \operatorname{Tr}_{BC}\left[\widetilde{\mathcal{F}}_i(\ketbras{00}_{BC} \otimes \rho_{AD})\right] \\
        &=  \sum_i c_i \mathcal{G}_i(\rho_{AD}).
    \end{align}
    The sampling overhead of the maximally entangled state using this QPD is $\sum_i |c_i|$, thus for the minimal sampling overhead $\hat{\gamma}^{\rho}(\Phi) \leq \sum_i |c_i|$ has to hold. 
    Since we assumed $\gamma^{\rho}(\mathcal{I}) = \sum_i |c_i|$, this implies $\gamma^{\rho}(\mathcal{I}) \ge \hat{\gamma}^{\rho}(\Phi)$.
    Consequently, this confirms the lower bound, thereby completing the proof.
\end{proof}

\section{Proof of \Cref{theorem_decomposition}}\label{sec:proof_decomposition}
\begin{proof}
As presented in \Cref{eq:tele2}, Pauli errors occur during the teleportation based on the overlap between the resource state and the Bell basis states.
The overlap between the NME state $\Phi^k$ and the Bell basis states is as follows:
\begin{align}
    \bra{\Phi^{I}}\Phi^k\ket{\Phi^{I}} &= \frac{\left(k + 1\right)^{2}}{2 \left(k^{2} + 1\right)}\label{eq:overlap_1}\\
    \bra{\Phi^{X}}\Phi^k\ket{\Phi^{X}} &= 0\label{eq:overlap_2}\\
    \bra{\Phi^{Y}}\Phi^k\ket{\Phi^{Y}} &= 0\label{eq:overlap_3}\\
    \bra{\Phi^{Z}}\Phi^k\ket{\Phi^{Z}} &= \frac{\left(k - 1\right)^2}{2 \left(k^{2} + 1\right)}\label{eq:overlap_4}
\end{align}
Therefore, the teleportation with $\Phi^k$ successfully transmits the state $\rho$ with probability $\bra{\Phi^{I}}\Phi^k\ket{\Phi^{I}}$ and exclusively introduces Pauli $Z$ errors with probability $\bra{\Phi^{Z}}\Phi^k\ket{\Phi^{Z}}$:
\begin{align}
    \mathcal{E}_{\text{tel}}^{\Phi^k}(\rho) = \sum_{\sigma \in\{I,Z\}}\bra{\Phi^{\sigma}}\Phi^k\ket{\Phi^{\sigma}}\sigma\rho\sigma.
\end{align}
Thus, the first summand in \Cref{eq:thrm2} with $U_1=H$ and $U_2=SH$ can be expressed as the following:
\begin{align}
    &\sum_{i\in\{1,2\}} U_i\mathcal{E}_{\text{tel}}^{\Phi^k}(U_i^{\dagger}\rho U_i)U_i^{\dagger}\\
    &= \sum_{i\in\{1,2\}}\sum_{\sigma \in\{I,Z\}}\bra{\Phi^{\sigma}}\Phi^k\ket{\Phi^{\sigma}}U_i \sigma U_i^{\dagger}\rho U_i \sigma U_i^{\dagger}\\
     \begin{split}
        &=2\bra{\Phi^{I}}\Phi^k\ket{\Phi^{I}}\rho 
        \\&\quad+ \sum_{i\in\{1,2\}}\bra{\Phi^{Z}}\Phi^k\ket{\Phi^{Z}}U_i Z U_i^{\dagger}\rho U_i Z U_i^{\dagger}
    \end{split}\\
    &=2\bra{\Phi^{I}}\Phi^k\ket{\Phi^{I}}\rho + \bra{\Phi^{Z}}\Phi^k\ket{\Phi^{Z}} (X \rho X + Y \rho Y), 
\end{align}
because 
\begin{align}
    U_1 Z U_1^{\dagger} = HZH =X \\
    U_2 Z U_2^{\dagger} = (SH)Z(SH)^{\dagger} = Y.
\end{align}
Using the previously calculated overlaps from \Cref{eq:overlap_1,eq:overlap_4}, this results in 
\begin{align}
    &\sum_{i\in\{1,2\}} U_i\mathcal{E}_{\text{tel}}^{\Phi^k}(U_i^{\dagger}\rho U_i)U_i^{\dagger} \\
    &\qquad= \frac{2\left(k + 1\right)^{2}}{2 \left(k^{2} + 1\right)} \rho
    + \frac{\left(k - 1\right)^2}{2 \left(k^{2} + 1\right)}
    (X \rho X + Y \rho Y).
\end{align}
Including the coefficient of the first summand in \Cref{eq:thrm2} results in
\begin{align}
    \begin{split}\label{eq:partial_channel}
    &\frac{k^2+1}{(k+1)^2}\sum_{i\in\{1,2\}} U_i\mathcal{E}_{\text{tel}}^{\Phi^k}(U_i^{\dagger}\rho U_i)U_i^{\dagger}\\
    &=\rho +  \frac{(k-1)^2}{(k+1)^2}\left(\frac{1}{2}X\rho X + \frac{1}{2}Y\rho Y\right).
    \end{split}
\end{align}

This weighted term already contains $\rho$.
However, the second summand represents the error introduced by using the NME state for teleportation.
To compute and correct this error term, we can express it in the following form:
\begin{align}
    &\frac{1}{2}X\rho X + \frac{1}{2}Y\rho Y \\
    \begin{split}
        &=\phantom{+} \frac{(\ketbra{0}{1} + \ketbra{1}{0})\rho(\ketbra{0}{1} + \ketbra{1}{0})}{2} \\
        &\quad+ \frac{(-i\ketbra{0}{1} + i\ketbra{1}{0})\rho(-i\ketbra{0}{1} + i\ketbra{1}{0})}{2}
    \end{split}\\
    &= \ketbra{0}{1}\rho\ketbra{1}{0} + \ketbra{1}{0}\rho\ketbra{0}{1}\\
    &= \braket{1|\rho|1}\ketbras{0} + \braket{0|\rho|0}\ketbras{1}\\
    &= \tr\left[\ketbras{1}\rho\right] \ketbras{0}
    + \tr\left[\ketbras{0}\rho\right] \ketbras{1}\\
    &= \sum_{j \in \{0,1\}}\tr\left[ \ketbras{j}\rho \right]X \ketbras{j}X.
\end{align}
This can be implemented by a single circuit that performs a measurement and initializes the inverse measurement result.

Using this representation of the error in \Cref{eq:partial_channel} and rearranging the equation shows
\begin{equation}
    \begin{split}
    \rho =&\phantom{-} \frac{k^2+1}{(k+1)^2}\sum_{i\in\{1,2\}} U_i\mathcal{E}_{\text{tel}}^{\Phi^k}(U_i^{\dagger}\rho U_i)U_i^{\dagger}\\
    &- \frac{(k-1)^2}{(k+1)^2}\sum_{j\in\{0,1\}} \tr\left[ \ketbras{j}\rho \right]X \ketbras{j}X,
    \end{split}
\end{equation}
which matches the QPD as outlined in \Cref{theorem_decomposition}.
It consists of two teleportation circuits and one measure-and-prepare circuit as depicted in \Cref{fig:cut_with_entanglement}.
Lastly, the sampling overhead associated with this decomposition, quantified by the sum of the absolute values of the coefficients of the circuits, is given by
\begin{align}
    &2 \frac{k^2+1}{(k+1)^2} + \frac{(k-1)^2}{(k+1)^2} \\
    &= 3 \frac{k^2+1}{(k+1)^2} - \frac{2k}{(k+1)^2} \\
    &= 4 \frac{k^2+1}{(k+1)^2} - \frac{(k+1)^2}{(k+1)^2} \\
    &= \frac{4(k^2+1)}{(k+1)^{2}} -1.
\end{align}
\end{proof}

\end{document}